\documentclass[11pt]{article}
\usepackage[utf8]{inputenc}
\usepackage[T1]{fontenc}
\usepackage{lmodern}
\usepackage[margin=1.3in]{geometry}
\usepackage{amsmath,amssymb,amsthm,amsfonts,mathtools}
\usepackage{graphicx}
\usepackage{xcolor}
\usepackage{tikz}
\usetikzlibrary{arrows.meta,positioning,calc}
\usepackage{enumitem}
\usepackage[colorlinks=true,linkcolor=blue,citecolor=blue,urlcolor=blue]{hyperref}
\usepackage{dsfont}
\usepackage{natbib}

\newtheorem{theorem}{Theorem}[section]
\newtheorem{proposition}[theorem]{Proposition}
\newtheorem{lemma}[theorem]{Lemma}
\newtheorem{corollary}[theorem]{Corollary}
\newtheorem{assumption}[theorem]{Assumption}
\theoremstyle{definition}

\theoremstyle{remark}
\newtheorem{remark}[theorem]{Remark}

\newcommand{\E}{\mathbb{E}}
\newcommand{\Prob}{\mathbb{P}}
\newcommand{\1}{{1}}
\newcommand{\dd}{\mathrm{d}}
\newcommand{\R}{\mathbb{R}}
\newcommand{\given}{\,\vert\,}
\newcommand{\Z}{\mathcal{Z}}
\newcommand{\Xset}{\mathcal{X}}
\newcommand{\Sset}{\mathcal{S}}
\newcommand{\as}{\xrightarrow{\mathrm{a.s.}}}

\newcommand{\dist}{\xrightarrow{\mathcal{D}}}

\title{Payment Process Estimation in Aggregated Insurance Models}

\author{
Martin Bladt$^{1}$ \and
Marcus Christiansen$^{2}$
}

\date{
\small
$^{1}$ Faculty of Mathematical Sciences, University of Copenhagen \\
$^{2}$ Institute of Mathematics, Carl von Ossietzky University of Oldenburg
}

\begin{document}
\maketitle

\begin{abstract}
Insurance payments may depend on latent micro states although only macro states and realized payments are observed. We study a sojourn-payment model for such aggregated multi-state systems under left-truncation and right-censoring. Starting from a micro-to-macro projection, we establish strong consistency and weak convergence for inverse-probability-weighted estimators of state-specific cumulative payment processes. 
\end{abstract}

\section{Introduction}\label{sec:intro}

In life and health insurance, it is common to explain the individual insurance payments from a multistate  model for the insured. In simple life insurances,  the multistate model can fully explain the insurance benefits, so it suffices to statistically estimate this multistate model only. However, if disability benefits or sickness benefits are covered, or even medical bills are paid by the insurer, then fully explanatory multistate models can be so highly complex that their statistical estimation is unfeasible. In this case, aggregated insurance models are commonly used, which  distinguish only a reduced set of macro states and average over the micro states. 

For example, consider German private health insurance policies, which are long-term contracts calculated in a similar way to life insurance policies. The true insurance benefits depend on the huge number of different health states that the insured person may experience, so it is extremely challenging to statistically estimate the micro-level transition rates. Instead, insurers only distinguish the macro states `active', `lapsed', and `dead', cf.~\citet{MilbrodtRoehra2016}.

A second example is disability protection, which, in many countries, is spread across several layers, including public benefits, occupational pension schemes and private insurance. The interplay between the various levels can be highly complex. This is why actuaries often aggregate different disability statuses into a single macro state, avoiding micro-level modelling. See, for example, \citet{FurrerSandqvist2025} for the situation in Denmark.

If the macro-level model does not fully explain the insurance payments process and the micro-level model is inaccessible, the aggregated insurance payment functions must be estimated directly from the data. This paper introduces a general framework for this statistical task and establishes the desirable asymptotic properties of the estimators.   The combination of micro-level dependent insurance payments and the unavailability of transition rates at the micro-level has not yet been adequately addressed in the existing literature.

\citet{AhmadBladtFurrer2023} and \citet{AhmadBladt2024} examine aggregated life insurance models in terms of the potential loss of the Markov property when switching from the micro- to the macro-level, in situations where the payment process depends solely on the macro states.  \citet{Andersen2024} study payment processes that  depend on the micro states, focusing on the computational efficiency of the calculation and assuming that the transition rates at the micro-level are known.
Micro-level modelling can be driven not only by the insured person’s health conditions, but also by the behaviour of the policyholder, cf.~\citet{MilhaudDutang2018}. \citet{ReckSchuppReuss2025} investigate algorithms for the efficient aggregation of behaviour-dependent micro-levels, discussing  the transition rates but not the payment processes.

This paper focuses on payment processes involving sojourn payments only. This reflects the structure of our real-life examples and simplifies the statistical problem. Future research should also consider transition payments. From a mathematical point of view, the empirical process techniques of~\citet{BladtFurrer2025} are the main tools upon which we base our analysis.  Earlier related work in survival analysis is vast, and we confine ourselves to highlight \citet{RamlauHansen1988} and \citet{LintonNielsen1995} for intensity estimation, and the general treatment of counting processes in~\citet{AndersenBorganGillKeiding2012}; the former references offer a promising route when considering the extension to more general payment processes.

The paper is structured as follows:
In Section \ref{sec:setup}, we introduce the modeling framework for  micro-level payment processes and prove the existence of unique projections  to the macro-level. Section \ref{sec:unconditional_framework} adds delayed entry and right-censoring to the payment processes. In Section \ref{sec:estimators} we introduce estimators for the macro level payment functions. We prove strong consistency in section \ref{subsec:strong_con}  and asymptotic normality in section \ref{subsec:weak_markov} under reasonable technical assumptions. Section \ref{sec:numerical_examples} provides a numerical example.

\section{Micro-level payment processes and macro-level projection}\label{sec:setup}

Insurance payments may depend on features finer than the contractual states used for premium calculation and reporting. 
We model this by separating latent micro states from observed macro states and then projecting the payment mechanism to macro level.
Thus, consider an individual insurance policy governed by a c\`adl\`ag jump process
\begin{align*}
Z=(X,S),\quad Z(t)\in\Z:=\Xset\times\Sset,
\end{align*}
on a finite state space. Here $X$ represents macro states and $S$ represents micro states. We assume that $Z(0-)=Z(0)$. 

We consider insurance payment processes on a finite time interval $[0,T]$. 
We assume that the actual payments between the insurer and the policyholder can be described by  a payment process $A$ on $[0,T]$ of the form
\begin{equation}\label{eq:A_decomp}
A(t)
=
\sum_{(i,e)\in\Z} \int_{[0,t]} \1_{\{Z(s-)=(i,e)\}}\,\dd A_{(i,e)}(s)
\end{equation}
with right-continuous functions $A_{(i,e)}$ of finite variation. Thus, payments only accumulate while the policy remains in a micro state. This is the setting studied throughout the paper. Our assumptions imply that $A$ has càdlàg paths of finite variation.

Since we focus on applications in which micro-level modeling is not really feasible, we are looking for a suitable projection of $A$ onto a macro-level process $B$ of the form
\begin{equation}\label{eq:B_decomp}
 B(t)
=
\sum_{i\in\Xset}\int_{[0,t]} \1_{\{X(s-)=i\}}\,\dd B_i(s)
\end{equation}
with right-continuous functions $B_i$ of finite variation. 

The following theorem shows that a suitable projection exists and is almost surely unique. The theorem needs the  state-specific mean cumulative payment functions
\begin{align}
	G_i(t)&:=\E\left[\int_{[0,t]}\1_{\{X(s-)=i\}}\dd A(s)\right]
	\label{eq:uncens_payments}
\end{align}
and the state occupation probabilities
\begin{align}
	W_i(t)&:=\E\left[\1_{\{X(t)=i\}}\right].
	\label{eq:uncens_stateoccupprob}
\end{align}

\begin{theorem}[Projection and uniqueness]\label{th:projection}
The macro-level payment process  $B$ defined by \eqref{eq:B_decomp}  and by the right-continuous finite-variation functions
\begin{align}\label{Bi_construction}
	B_i(t) = \int_{[0,t]} \frac{1}{W_i(s-)} G_i(\dd s)
\end{align}
 satisfies the projection property 
 \begin{equation}\label{eq:projection_condition}
 	\E\bigg[ \int_I\1_{\{X(s-)=i\}}\,\dd A(s)\bigg]
 	=
 	\E\bigg[\int_I \1_{\{X(s-)=i\}}\,\dd B(s)\bigg]
 \end{equation}
 for all  intervals  $I \subset [0,T]$ and states $i \in \mathcal{X}$.
 
If $\widetilde B$ is another macro payment process of the form \eqref{eq:B_decomp} satisfying \eqref{eq:projection_condition}, then $\widetilde B=B$ on $[0,T]$  up to indistinguishability.
\end{theorem}
The projection property \eqref{eq:projection_condition} is similar to the projection property  of 
\citet[Theorem 2]{Andersen2024}.  We additionally show uniqueness and do not require Markovianity. However, unlike \citet{Andersen2024}, we do not consider extended filtrations.
The nonnegative Lebesgue–Stieltjes integral in \eqref{Bi_construction} uses the standard convention that $0/0 := 0$. This convention is used throughout the paper.

\begin{proof}
	Using Fubini's theorem, the definition \eqref{Bi_construction} can be equivalently written as 
\begin{align}\begin{split}
		\label{ExplConstrOfB}
		 B_{i}(t) &= \sum_{e \in \mathcal{S}} \int_{[0,t]} \frac{ \mathbb{E}[   \1_{\{Z(s-)=(i,e)\}} ]}{\sum_{e' \in \mathcal{S}} \mathbb{E}[\1_{\{Z(s-)=(i,e')\}}  ]}   \dd A_{(i,e)}(s).
\end{split}\end{align}
 Note that the non-negative numerator is always less than or equal to the non-negative denominator. This means that the functions $B_i$ inherit the finite-variation property of the functions $A_{(i,e)}$.  In particular,  the numerator is always zero when the denominator is zero, which ensures that the non-negative Lebesgue-Stieltjes integral in \eqref{Bi_construction} is well-defined even when the denominator is zero.
 
 Applying Fubini's theorem, we get
\begin{align}\begin{split}\label{EquilityofEdBEdA}
	\mathbb{E}[ \1_{\{X(t-)=i\}} \dd B (t) ]
	&= \mathbb{E}[ \1_{\{X(t-)=i\}}] \dd B_i (t) \\
	& =  \sum_{e \in \mathcal{S}} \mathbb{E}[   \1_{\{Z(t-)=(i,e)\}} ] \dd A_{(i,e)}(t)\\
	&=\mathbb{E}[ \1_{\{X(t-)=i\}} \dd A (t) ],
\end{split}\end{align}
which confirms that \eqref{ExplConstrOfB}  satisfies the projection property \eqref{eq:projection_condition}.

Suppose that $\tilde{B}$ is another macro-level payment process  that satisfies \eqref{eq:projection_condition}. 	 
Then, for 
any random time $\tau:\Omega \rightarrow [0,T]$, it holds that 
\begin{align*}
	\E[ B(\tau) ] &= \E\bigg[\sum_{i\in \mathcal{X}} \int_{[0,T]}    \1_{\{\tau \geq t\}} \1_{\{X(t-)=i\} } \dd B_i(t) \bigg]\\
	&= \sum_{i\in \mathcal{X}}  \int_{[0,T]} \frac{\E[   \1_{\{\tau \geq t\}} \1_{\{X(t-)=i\} }]}{\E[   \1_{\{X(t-)=i\} }]} \E[    \1_{\{X(t-)=i\} } \dd B(t) ] \\
	&= \sum_{i\in \mathcal{X}}  \int_{[0,T]} \frac{\E[   \1_{\{\tau \geq t\}} \1_{\{X(t-)=i\} }]}{\E[   \1_{\{X(t-)=i\} }]} \E[    \1_{\{X(t-)=i\} } \dd \tilde B(t) ] \\
		&= \E[ \tilde B(\tau) ],
\end{align*}
using the convention $0/0:=0$.
This means that the right-continuous processes $B$ and $\tilde B$ are almost surely equal.
\end{proof}
\begin{remark}
	Different from \eqref{eq:A_decomp}, 
in the insurance literature we also find sojourn payment processes on $[0,T]$ of the form 
\begin{equation*}\dd A(t)
	=
	\sum_{(i,e)\in\Z}\1_{\{Z(t)=(i,e)\}}\,\dd A_{(i,e)}(t).
\end{equation*}
The difference is here that the payment $\dd A_{(i,e)}(t)$ is triggered by the event $\{Z(t)=(i,e)\}$ instead of $\{Z(t-)=(i,e)\}$. All the results in this paper can be transferred to this alternative setting. Throughout the paper, we would  we need to replace  $t-$ by $t$. 
\end{remark}

Theorem \ref{th:projection} fixes the macro target induced by the micro model.
The next section adds delayed entry and right-censoring through an observation-window process.

\section{Limited observation window}\label{sec:unconditional_framework}




We introduce a latent observation window $(L,R)$, where $L$ is a delayed entry time and $R$ represents right-censoring.  We do not require direct observation of $R$ itself, it is enough that the observation windows is observed up to $R$ or termination of the insurance  contract. 

The following assumption isolates the observation window from the macro-level multistate process and the micro-level payment process. 

\begin{assumption}[Independent observation window]\label{ass:indep_censoring}
Let $(L,R)$  be stochastically indepedent of $(X,A)$, and let  $\Prob( L < t \leq R)>0$ for all $t \in [0,T]$.
\end{assumption}

%
%
%
%
In addition to \eqref{eq:uncens_payments} and \eqref{eq:uncens_stateoccupprob}, we define the  censored means
\begin{align}
G_i^{\mathrm c}(t)&:=\E\left[\int_{[0,t]}\1_{\{L < s \leq R\}}\1_{\{X(s-)=i\}}\dd A(s)\right]
\label{eq:cens_payments}
\end{align}
and
\begin{align}
	W^{\mathrm c}_i(t)&:=\E\left[\1_{\{L \leq  t < R\}}\1_{\{X(t)=i\}}\right].
	\label{eq:cens_stateoccupprob}
\end{align}
Under Assumption~\ref{ass:indep_censoring}, censored and uncensored payment means differ only by a predictable factor. This allows us to write the target \eqref{Bi_construction} in terms of the  censored means.

\begin{proposition}[Substitution formula]\label{prop:alt_rep}
Under Assumption~\ref{ass:indep_censoring}, it holds that
\begin{align}
\int_{[0,t]}\frac{1}{W^{\mathrm c}_i(s-)}G_i^{\mathrm c}(\dd s)&=\int_{[0,t]}\frac{1}{W_i(s-)}G_i(\dd s).
\label{eq:G_relation}
\end{align}
%
\end{proposition}
\begin{proof}
Fix $t\le T$, and let $\Phi$ be a nonnegative random measure  with $\sigma(\Phi)\subseteq \sigma(X,A)$ and $\E[\Phi([0,t])]<\infty$.  Then
\begin{align*}
\E\left[\int_{[0,t]}\1_{\{L <  s \leq  R\}}\Phi(\dd s)\right]
&=\E\left[\E\left[\int_{[0,t]}\1_{\{L <  s \leq  R\}}\Phi(\dd s)\bigg|X,A\right]\right] \\
&=
\E\left[\int_{[0,t]}\E\left[\1_{\{L <  s \leq  R\}}\given X,A\right]\Phi(\dd s)\right] \\
&=
\E\left[\int_{[0,t]}\E[\1_{\{L <  s \leq  R\}}]\Phi(\dd s)\right],
\end{align*}
where the second equality uses the conditional Tonelli theorem for nonnegative integrands, and the third uses the stochastic independence of $(L,R)$ and $(X,A)$. 
Now define the mean measure $\nu_\Phi(B):=\E[\Phi(B)]$. For nonnegative simple functions $f=\sum_{m=1}^M a_m\1_{B_m}$,
\begin{align*}
\E\left[\int_{[0,t]}f(s)\,\Phi(\dd s)\right]
&=
\sum_{m=1}^M a_m\,\E[\Phi(B_m)] \\
&=
\sum_{m=1}^M a_m\,\nu_\Phi(B_m) \\
&=
\int_{[0,t]}f(s)\,\nu_\Phi(\dd s).
\end{align*}
By monotone convergence, this extends to all nonnegative measurable functions $f$. Taking $f(s)=\E[\1_{\{L <  s \leq  R\}}] $ yields
\begin{align*}
\E\left[\int_{(0,t]}\E[\1_{\{L <  s \leq  R\}}]\Phi(\dd s)\right]
=
\int_{(0,t]}\E[\1_{\{L <  s \leq  R\}}]\,\nu_\Phi(\dd s)
=
\int_{(0,t]}\E[\1_{\{L <  s \leq  R\}}]\,\E[\Phi(\dd s)].
\end{align*}
We apply this with $\Phi(\dd s)=\1_{\{X(s-)=i\}}\dd A(s)$ to get
\begin{align*}
	G_i^{\mathrm c} (\dd s) =  \E[\1_{\{L <  s \leq  R\}}]\, G_i (\dd s).
\end{align*}
 The stochastic independence of $(L,R)$ and $(X,A)$ gives
\begin{align*}
	W_i^{\mathrm c} (s-) =  \E[\1_{\{L <  s \leq  R\}}]\, W_i(s-).
\end{align*}
The latter two equations yield \eqref{eq:G_relation} since $\E[\1_{\{L <  s \leq  R\}}]>0$ by 
Assumption \ref{ass:indep_censoring}.
\end{proof}

Proposition~\ref{prop:alt_rep} is the basis for estimation:  the payment functions $B_i$, $i \in \mathcal{X}$, can be recovered from observable censored payment sums and censored state occupation counts.
%
%
The next section turns this identity into an estimator and studies its large-sample behavior.

\section{Estimators}\label{sec:estimators}

We assume that we have a stochastically independent and identically distributed sample
\begin{align*}
(X^p,A^p,L^p,R^p)_{p=1}^n.
\end{align*}
This is a probabilistic representation, not an observation requirement. The estimators below only use observable stopped portfolio summaries and state  occupation counts. 

We define the empirical analogues of the population objects from Section~\ref{sec:unconditional_framework}:
\begin{align}
\mathbb W_i^{(n)}(t)
&:=\frac1n\sum_{p=1}^n\1_{\{L^p \leq  t <  R^p\}}\1_{\{X^p(t)=i\}},
\label{eq:W_emp}
\\
\mathbb G_i^{(n)}(t)
&:=\frac1n\sum_{p=1}^n\int_{[0,t]}\1_{\{L^p < s \leq  R^p\}}\1_{\{X^p(s-)=i\}}\dd A^p(s).
\label{eq:G_emp}
\end{align}
Based on these, we define the estimator for the target:
\begin{align}
\mathbb B_i^{(n)}(t)
&:=\int_{[0,t]}\frac{1}{\mathbb W_i^{(n)}(s-)}\,\mathbb G_i^{(n)}(\dd s).
\label{eq:BBhat}
\end{align}
 This is the direct empirical analogue of \eqref{eq:G_relation}.

To compute $\mathbb B_i^{(n)}$, one does not need the full collection of individual payment paths or the latent right-censoring times $R^p$ themselves. For each time $t$, one must observe only the sums
\begin{align*}
\sum_{p=1}^n\1_{\{L^p <  t \leq  R^p\}} \1_{\{X^p(t-)=i\}}
\end{align*}
and 
\begin{align*}
\sum_{p=1}^n\int_{(0,t]}\1_{\{L^p < s \leq  R^p\}}\1_{\{X^p(s-)=i\}}\dd A^p(s).
\end{align*}
These quantities are naturally available in insurance  data systems.  Therefore, the estimator \eqref{eq:BBhat} is usually computable in insurance practice.


We now introduce technical assumptions to ensure the estimators have  desirable asymptotic properties.
%

\begin{assumption}[Moment condition for payment process]\label{ass:markov_moments}
	For some $\delta>0$, let
	\begin{align*}
		\E\bigg[\bigg(\int_{[0,T]} |\dd A(u)| \bigg)^{2+\delta}\bigg]<\infty.
	\end{align*}
\end{assumption}
Note that the integral in the latter assumption represents the total variation of $A$ on $[0,T]$. 
If is $A$ monotone, then we can replace this integral simply by  $|A(T)|$.

For all $i,j \in \mathcal{X}$ with $i \neq j$, we define the counting processes
\begin{align*}
	N_{ij}(t):=\sharp\big\{ s \in (0,t]: X(s-)=i,\ X(s)=j\big\},
\end{align*}
giving the number of jumps of $X$ from $i$ to $j$ on the time interval $(0,t]$. 
\begin{assumption}[Moment condition for macro-level transition counts]\label{ass:boundforjumpcounter}
For all $i,j \in \mathcal{X}$ with $i\neq j$, let
$$\E[N_{ij}(T)^2]<\infty.$$
\end{assumption}
For the consistency result, first moments would suffice. We impose second moments already here because the same condition is needed for the weak-convergence arguments below.

\begin{assumption}[Positivity]\label{ass:markov_pos}
	For each $i \in \mathcal{X}$, let $[\alpha_i,\beta_i]\subset [0,T]$ be an interval such that
	\begin{align*}
		\inf_{t \in [\alpha_i,\beta_i] }W_i^{\mathrm c}(t)>0.
	\end{align*}
\end{assumption}
This assumption specifies intervals on which inverse weighting is well-defined.

\section{Strong consistency of estimators}\label{subsec:strong_con}

We now show strong consistency of the estimators $\mathbb B_i^{(n)}$, $i \in \mathcal{X}$, on the intervals $[\alpha_i,\beta_i]$, $i \in \mathcal{X}$, from Assumption \ref{ass:markov_pos}.

\begin{theorem}[Strong uniform consistency]\label{th:strong_cons}
	Under Assumptions~\ref{ass:indep_censoring}, \ref{ass:markov_moments}, \ref{ass:boundforjumpcounter},  and~\ref{ass:markov_pos},
	\begin{align}
		\sup_{t \in [\alpha_i,\beta_i]}\left|
		\big(\mathbb B_i^{(n)}(t)-\mathbb B_i^{(n)}(\alpha_i)\big)
		-\big(B_i(t)-B_i(\alpha_i)\big)\right|\as0.
		\label{eq:cons_BB}
	\end{align}
\end{theorem}

Before we proof this theorem, we show uniform convergence of the emprical processes $ \mathbb G_i^{(n)}$ and $ \mathbb W_i^{(n)}$.

\begin{proposition}[Uniform consistency of components]\label{prop:comp_cons}
Under Assumptions~\ref{ass:indep_censoring}, ~\ref{ass:markov_moments}, and \ref{ass:boundforjumpcounter},
\begin{align}
\sup_{0\le t\le T}\big|\mathbb G_i^{(n)}(t)-G_i^{\mathrm c}(t)\big|&\as0,
\label{eq:G_cons}
\\
\sup_{0\le t\le T}\big|\mathbb W_i^{(n)}(t)-W^{\mathrm c}_i(t)\big|&\as0.
\label{eq:W_cons}
\end{align}
\end{proposition}
\begin{proof}
Since $A$ has càdlàg paths of finite variation, it can be decomposed into the difference of two monotone càdlàg processes. Based on this fact and the triangular inequality, it is sufficient to prove convergence for non-decreasing payment processes $A$ only.

The following key decomposition follows from counting the number of jumps into and out of state $i \in \mathcal{X}$ under left truncation and right censoring:
\begin{align}
\1_{\{L\le t< R\}}\1_{\{X(t)=i\}}
&=
\1_{\{L\le t\}}\1_{\{L<R\}}\1_{\{X(L)=i\}}
-\1_{\{R\le t\}}\1_{\{L<R\}}\1_{\{X(R\wedge T)=i\}}
\nonumber\\
&\quad
+\sum_{j:j\neq i}\int_{(0,t]}\1_{\{L<s\le R\}}\,\dd N_{ji}(s)
-\sum_{j:j\neq i}\int_{(0,t]}\1_{\{L<s\le R\}}\,\dd N_{ij}(s).
\label{eq:key_decomp}
\end{align}
Indeed, if $t<L$, then both sides are zero. If $L\le t<R$, then
\begin{align*}
\1_{\{X(t)=i\}}
=
\1_{\{X(L)=i\}}
+\sum_{j:j\neq i}\int_{(L,t]}\dd N_{ji}(s)
-\sum_{j:j\neq i}\int_{(L,t]}\dd N_{ij}(s),
\end{align*}
since the state indicator can change only through jumps into and out of the state $i$. Multiplying by $\1_{\{L<R\}}$ gives \eqref{eq:key_decomp} on $[L,R)$. If $t\ge R$, then the same identity evaluated at time $R\wedge T$ shows that the right-hand side equals $\1_{\{L<R\}}\1_{\{X(R\wedge T)=i\}}-\1_{\{L<R\}}\1_{\{X(R\wedge T)=i\}}=0$, which matches the left-hand side.

Define the auxiliary empirical processes
\begin{align*}
\mathbb J_i^{(n)}(t)
&:=\frac1n\sum_{p=1}^n\1_{\{L^p\le t\}}\1_{\{L^p<R^p\}}\1_{\{X^p(L^p)=i\}},\\
\mathbb C_i^{(n)}(t)
&:=\frac1n\sum_{p=1}^n\1_{\{R^p\le t\}}\1_{\{L^p<R^p\}}\1_{\{X^p(R^p\wedge T)=i\}},\\
\mathbb N_{ji}^{\mathrm c,(n)}(t)
&:=\frac1n\sum_{p=1}^n\int_{(0,t]}\1_{\{L^p<s\le R^p\}}\,\dd N^p_{ji}(s),
\qquad j\neq i,
\end{align*}
and their population counterparts
\begin{align*}
J_i(t)
&:=\Prob(L\le t,\ L<R,\ X(L)=i),\\
C_i(t)
&:=\Prob(R\le t,\ L<R,\ X(R\wedge T)=i),\\
N_{ji}^{\mathrm c}(t)
&:=\E\left[\int_{(0,t]}\1_{\{L<s\le R\}}\,\dd N_{ji}(s)\right],
\qquad j\neq i.
\end{align*}
Averaging \eqref{eq:key_decomp} empirically and in expectation gives
\begin{align}
\mathbb W_i^{(n)}(t)
&=
\mathbb J_i^{(n)}(t)-\mathbb C_i^{(n)}(t)
+\sum_{j:j\neq i}\Big(\mathbb N_{ji}^{\mathrm c,(n)}(t)-\mathbb N_{ij}^{\mathrm c,(n)}(t)\Big),
\label{eq:key_decomp_emp}
\\
W_i^{\mathrm c}(t)
&=
J_i(t)-C_i(t)
+\sum_{j:j\neq i}\Big(N_{ji}^{\mathrm c}(t)-N_{ij}^{\mathrm c}(t)\Big).
\label{eq:key_decomp_pop}
\end{align}

For fixed $t$, the random variables
\begin{align*}
\mathbb G_i^{(n)}(t),\quad \mathbb J_i^{(n)}(t),\quad \mathbb C_i^{(n)}(t),\quad \mathbb N_{ji}^{\mathrm c,(n)}(t),\quad \mathbb N_{ij}^{\mathrm c,(n)}(t)
\end{align*}
are empirical means of iid integrable random variables. Therefore, by the strong law of large numbers,
\begin{align*}
\mathbb G_i^{(n)}(t)&\as G_i^{\mathrm c}(t),\\
\mathbb J_i^{(n)}(t)&\as J_i(t),\\
\mathbb C_i^{(n)}(t)&\as C_i(t),\\
\mathbb N_{ji}^{\mathrm c,(n)}(t)&\as N_{ji}^{\mathrm c}(t),\qquad j\neq i.
\end{align*}
The same argument applies to the left limits.

Each map
\begin{align*}
t\mapsto \mathbb G_i^{(n)}(t),\quad
t\mapsto \mathbb J_i^{(n)}(t),\quad
t\mapsto \mathbb C_i^{(n)}(t),\quad
t\mapsto \mathbb N_{ji}^{\mathrm c,(n)}(t)
\end{align*}
is nondecreasing and c\`adl\`ag. Each corresponding limit map is a finite-measure cumulative function on $[\alpha_i,\beta_i]$. Therefore Theorem~\ref{th:abstract_gc_cadlag} applies to each component sequence, yielding uniform almost sure convergence. Here, the required integrability is immediate for $\mathbb J_i^{(n)}$ and $\mathbb C_i^{(n)}$, while
\begin{align*}
\mathbb N_{ji}^{\mathrm c,(n)}(\theta)\le \frac1n\sum_{p=1}^n N_{ji}^p(T)
\end{align*}
and $N_{ji}(T)\in L^1(\Prob)$ by Assumption~\ref{ass:boundforjumpcounter}.

Using \eqref{eq:key_decomp_emp} and \eqref{eq:key_decomp_pop}, together with the triangle inequality and the uniform convergences above, we obtain
\begin{align*}
\sup_{0\le t\le T}\big|\mathbb W_i^{(n)}(t)-W_i^{\mathrm c}(t)\big|
&\le
\sup_{0\le t\le T}\big|\mathbb J_i^{(n)}(t)-J_i(t)\big|
+\sup_{0\le t\le T}\big|\mathbb C_i^{(n)}(t)-C_i(t)\big|
\\
&\quad
+\sum_{j:j\neq i}\sup_{0\le t\le T}\big|\mathbb N_{ji}^{\mathrm c,(n)}(t)-N_{ji}^{\mathrm c}(t)\big|
+\sum_{j:j\neq i}\sup_{0\le t\le T}\big|\mathbb N_{ij}^{\mathrm c,(n)}(t)-N_{ij}^{\mathrm c}(t)\big|\\
& 
\as0.
\end{align*}
This proves \eqref{eq:G_cons} and \eqref{eq:W_cons}.
\end{proof}

We now transfer the component convergences to the inverse-weighted payment function estimator. Continuity and Hadamard differentiability results for the ratio-integral map are collected in Appendix~\ref{app:ep_tools}; see Lemma~\ref{lem:hadamard} and Corollary~\ref{lem:map_cont}.

\begin{proof}[Proof of Theorem \ref{th:strong_cons}]
By Proposition~\ref{prop:comp_cons},
\begin{align*}
(\mathbb W_i^{(n)},\mathbb G_i^{(n)})\as(W^{\mathrm c}_i,G_i^{\mathrm c})
\quad\text{uniformly on }[\alpha_i,\beta_i].
\end{align*}
Apply Corollary~\ref{lem:map_cont} with
\begin{align*}
(x_n,y_n)=(\mathbb W_i^{(n)},\mathbb G_i^{(n)}),
\quad
(x,y)=(W_i^{\mathrm c},G_i^{\mathrm c}).
\end{align*}
Then, using Assumption~\ref{ass:markov_pos},
\begin{align*}
\sup_{t \in [\alpha_i,\beta_i]}\left|
\big(\mathbb B_i^{(n)}(t)-\mathbb B_i^{(n)}(\alpha_i)\big)
-\int_{(\alpha_i,t]}\frac{1}{W_i^{\mathrm c}(s-)}\,G_i^{\mathrm c}(\dd s)\right|\as0.
\end{align*}
By Proposition~\ref{prop:alt_rep}, the limit integral equals $B_i(t)-B_i(\alpha_i)$.
\end{proof}

\section{Asymptotic normality}\label{subsec:weak_markov}

We now show asymptotic normality of the estimators $\mathbb B_i^{(n)}$, $i \in \mathcal{X}$, on the intervals $[\alpha_i,\beta_i]$, $i \in \mathcal{X}$, from Assumption \ref{ass:markov_pos}. Before we do that, we  derive the Gaussian limits of the empirical components  
\begin{align*}
	\xi_{i,G}^{(n)}(t)&:=\sqrt{n}\big(\mathbb G_i^{(n)}(t)-G_i^{\mathrm c}(t)\big),
	\\
	\xi_{i,W}^{(n)}(t)&:=\sqrt{n}\big(\mathbb W_i^{(n)}(t)-W_i^{\mathrm c}(t)\big).
\end{align*}

\begin{proposition}[Weak convergence of components]\label{th:weak_comp_markov}
	Under Assumptions~\ref{ass:indep_censoring}, ~\ref{ass:markov_moments}, and~\ref{ass:boundforjumpcounter}, the stacked process
	\begin{align*}
		\big(\xi_{i,G}^{(n)},\xi_{i,W}^{(n)}\big)
	\end{align*}
	converges weakly in $\ell^\infty([0,\theta])^2$ to a centered tight Gaussian process
	\begin{align*}
		\big(\xi_{i,G},\xi_{i,W}\big),
	\end{align*}
	with covariance functions, for $s,t\in[0,\theta]$,
	\begin{align*}
		\operatorname{Cov}(\xi_{i,G}(s),\xi_{i,G}(t))
		&=
		\E\left[
		\left(\int_{[0,s]}\1_{\{L<u\le R\}}\1_{\{X(u-)=i\}}\,\dd A(u)\right)
		\left(\int_{[0,t]}\1_{\{L<v\le R\}}\1_{\{X(v-)=i\}}\,\dd A(v)\right)
		\right]\\
		&\quad-G_i^{\mathrm c}(s)G_i^{\mathrm c}(t),
		\\
		\operatorname{Cov}(\xi_{i,G}(s),\xi_{i,W}(t))
		&=
		\E\left[
		\left(\int_{[0,s]}\1_{\{L<u\le R\}}\1_{\{X(u-)=i\}}\,\dd A(u)\right)
		\1_{\{L\le t<R\}}\1_{\{X(t)=i\}}
		\right]\\
		&\quad-G_i^{\mathrm c}(s)W_i^{\mathrm c}(t),
		\\
		\operatorname{Cov}(\xi_{i,W}(s),\xi_{i,W}(t))
		&=
		\Prob(L\le s<R,\ X(s)=i,\ L\le t<R,\ X(t)=i)
		\\
		&\quad -W_i^{\mathrm c}(s)W_i^{\mathrm c}(t).
	\end{align*}
\end{proposition}
\begin{theorem}[Weak convergence of the payment estimator]\label{th:weak_na_markov}
	Under Assumptions~\ref{ass:indep_censoring}, \ref{ass:markov_moments}, \ref{ass:boundforjumpcounter}, and~\ref{ass:markov_pos},
	\begin{align*}
		\sqrt{n}\Big((\mathbb B_i^{(n)}-\mathbb B_i^{(n)}(\alpha_i))-(B_i-B_i(\alpha_i))\Big)\dist\zeta_i
		\quad\text{in }\ell^\infty([\alpha_i,\beta_i]),
	\end{align*}
	where
	\begin{equation}\label{eq:weak_B_repr}
		\zeta_i(t)
		=
		\int_{(\alpha_i,t]}\frac{1}{W_i^{\mathrm c}(s-)}\xi_{i,G}(\dd s)
		-
		\int_{(\alpha_i,t]}\frac{\xi_{i,W}(s-)}{\big(W_i^{\mathrm c}(s-)\big)^2}\,G_i^{\mathrm c}(\dd s).
	\end{equation}
\end{theorem}
\begin{proof}[Proof of Proposition \ref{th:weak_comp_markov}]
Let $O=(X,A,L,R)$ denote one observation. Define the monotone cumulative responses
\begin{align*}
\psi_{i,G,t}(O)&:=\int_{[0,t]}\1_{\{L<s\le R\}}\1_{\{X(s-)=i\}}\,\dd A(s),\\
\psi_{i,J,t}(O)&:=\1_{\{L\le t\}}\1_{\{L<R\}}\1_{\{X(L)=i\}},\\
\psi_{i,C,t}(O)&:=\1_{\{R\le t\}}\1_{\{L<R\}}\1_{\{X(R\wedge T)=i\}},\\
\psi_{ji,N,t}(O)&:=\int_{(0,t]}\1_{\{L<s\le R\}}\,\dd N_{ji}(s),\qquad j\neq i.
\end{align*}
For fixed $O$, each map $t\mapsto \psi_{\bullet,t}(O)$ above is c\`adl\`ag and nondecreasing. Their envelopes are
\begin{align*}
F_{i,G}=A(T),\qquad F_{i,J}=1,\qquad F_{i,C}=1,\qquad F_{ji,N}=N_{ji}(T),
\end{align*}
which all belong to $L^2(\Prob)$ by Assumptions~\ref{ass:markov_moments} and~\ref{ass:boundforjumpcounter}.

For each class
\begin{align*}
\mathcal F_{i,G}&:=\{\psi_{i,G,t}:t\in[0,\theta]\},\\
\mathcal F_{i,J}&:=\{\psi_{i,J,t}:t\in[0,\theta]\},\\
\mathcal F_{i,C}&:=\{\psi_{i,C,t}:t\in[0,\theta]\},\\
\mathcal F_{ji,N}&:=\{\psi_{ji,N,t}:t\in[0,\theta]\},\qquad j\neq i,
\end{align*}
pointwise measurability follows from right-continuity. By Lemma~\ref{lem:mono_bracketing}, each class has bracketing entropy
\begin{align*}
N_{[\,]}\left(\varepsilon\|F_{\bullet}\|_{L^2(\Prob)},\mathcal F_{\bullet},L^2(\Prob)\right)\le \frac{C_{\bullet}}{\varepsilon^2},
\end{align*}
hence finite entropy integral. Therefore Theorem~\ref{Donsker} yields weak convergence in $\ell^\infty([0,\theta])$ for
\begin{align*}
\xi_{i,G}^{(n)},\quad
\xi_{i,J}^{(n)}(t)&:=\sqrt n\big(\mathbb J_i^{(n)}(t)-J_i(t)\big),\\
\xi_{i,C}^{(n)}(t)&:=\sqrt n\big(\mathbb C_i^{(n)}(t)-C_i(t)\big),\\
\xi_{ji,N}^{(n)}(t)&:=\sqrt n\big(\mathbb N_{ji}^{\mathrm c,(n)}(t)-N_{ji}^{\mathrm c}(t)\big),\qquad j\neq i.
\end{align*}

To obtain joint convergence, fix finite collections of times and components from the family
\begin{align*}
\xi_{i,G}^{(n)},\quad \xi_{i,J}^{(n)},\quad \xi_{i,C}^{(n)},\quad \xi_{ji,N}^{(n)},\qquad j\neq i.
\end{align*}
The resulting finite-dimensional vectors are empirical means of iid centered random vectors with finite second moments, so the multivariate CLT gives convergence of all mixed finite-dimensional distributions. Lemma~\ref{stack_convergence} then yields joint weak convergence of the full stack.

Now use the exact decompositions \eqref{eq:key_decomp_emp} and \eqref{eq:key_decomp_pop} to obtain
\begin{align*}
\xi_{i,W}^{(n)}(t)
=
\xi_{i,J}^{(n)}(t)-\xi_{i,C}^{(n)}(t)
+\sum_{j:j\neq i}\Big(\xi_{ji,N}^{(n)}(t)-\xi_{ij,N}^{(n)}(t)\Big).
\end{align*}
Hence $\xi_{i,W}^{(n)}$ is a continuous linear transform of the already convergent stack, so it converges weakly in $\ell^\infty([0,\theta])$. Applying the corresponding continuous linear map to the full stack and then projecting onto
\begin{align*}
\big(\xi_{i,G}^{(n)},\xi_{i,W}^{(n)}\big)
\end{align*}
gives the claimed joint limit.

For the covariance functions, note that
\begin{align*}
\xi_{i,W}^{(n)}(t)
=
\frac{1}{\sqrt n}\sum_{p=1}^n\Big(\1_{\{L^p\le t<R^p\}}\1_{\{X^p(t)=i\}}-W_i^{\mathrm c}(t)\Big),
\end{align*}
so the one-observation response associated with $\xi_{i,W}^{(n)}(t)$ is
\begin{align*}
\psi_{i,W,t}(O):=\1_{\{L\le t<R\}}\1_{\{X(t)=i\}}.
\end{align*}
Gaussian empirical process limits satisfy
\begin{align*}
\operatorname{Cov}(Gf,Gg)=\operatorname{Cov}(f(O),g(O)).
\end{align*}
Applying this identity with $(f,g)=(\psi_{i,G,s},\psi_{i,G,t})$ gives the first covariance function. Applying it with $(f,g)=(\psi_{i,G,s},\psi_{i,W,t})$ gives the second. Finally, applying it with $(f,g)=(\psi_{i,W,s},\psi_{i,W,t})$ gives
\begin{align*}
\operatorname{Cov}(\xi_{i,W}(s),\xi_{i,W}(t))
&=
\operatorname{Cov}\big(\1_{\{L\le s<R\}}\1_{\{X(s)=i\}},\1_{\{L\le t<R\}}\1_{\{X(t)=i\}}\big)\\
&=
\Prob(L\le s<R,\ X(s)=i,\ L\le t<R,\ X(t)=i)
-W_i^{\mathrm c}(s)W_i^{\mathrm c}(t),
\end{align*}
which is the third displayed covariance function.
\end{proof}

\begin{proof}[Proof of Theorem \ref{th:weak_na_markov}]
By Proposition~\ref{th:weak_comp_markov},
\begin{align*}
\sqrt{n}\big((\mathbb W_i^{(n)},\mathbb G_i^{(n)})-(W_i^{\mathrm c},G_i^{\mathrm c})\big)
\dist
(\xi_{i,W},\xi_{i,G})
\quad\text{in }\ell^\infty([0,\theta])^2.
\end{align*}
Set
\begin{align*}
\eta_i:=\inf_{t_i\le t\le\theta}W_i^{\mathrm c}(t)>0,
\end{align*}
which is strictly positive by Assumption~\ref{ass:markov_pos}. By Proposition~\ref{prop:comp_cons},
\begin{align*}
\Prob\big(\inf_{t_i\le t\le\theta}\mathbb W_i^{(n)}(t)\ge\eta_i/2\big)\to1.
\end{align*}
Hence, with probability tending to one, the pair $(\mathbb W_i^{(n)},\mathbb G_i^{(n)})$ belongs to the domain
\begin{align*}
\mathbb D_{\eta_i/2}
:=
\Big\{(x,y)\in D([t_i,\theta])^2:\inf_{t_i\le t\le\theta}x(t)\ge \eta_i/2,\ y\in BV([t_i,\theta])\Big\}.
\end{align*}

Define
\begin{align*}
\Psi_i(x,y)(t):=\int_{(t_i,t]}\frac{1}{x(s-)}\,y(\dd s),\quad t\in[t_i,\theta].
\end{align*}
By Lemma~\ref{lem:hadamard}, $\Psi_i$ is Hadamard differentiable at $(W_i^{\mathrm c},G_i^{\mathrm c})$, tangentially to $D([t_i,\theta])^2$, with derivative
\begin{align*}
\Psi'_{i,(W_i^{\mathrm c},G_i^{\mathrm c})}(h_1,h_2)(t)
=
\int_{(t_i,t]}\frac{1}{W_i^{\mathrm c}(s-)}\,h_2(\dd s)
-
\int_{(t_i,t]}\frac{h_1(s-)}{\big(W_i^{\mathrm c}(s-)\big)^2}\,G_i^{\mathrm c}(\dd s).
\end{align*}
Therefore Theorem~\ref{delta_method} yields
\begin{align*}
\sqrt n\left(\Psi_i(\mathbb W_i^{(n)},\mathbb G_i^{(n)})-\Psi_i(W_i^{\mathrm c},G_i^{\mathrm c})\right)
\dist
\Psi'_{i,(W_i^{\mathrm c},G_i^{\mathrm c})}(\xi_{i,W},\xi_{i,G}),
\end{align*}
in $\ell^\infty([t_i,\theta])$. By the displayed derivative formula, the limit on the right-hand side is exactly the process \eqref{eq:weak_B_repr}.

It remains to identify the two arguments of $\Psi_i$. By definition of the estimator,
\begin{align*}
\Psi_i(\mathbb W_i^{(n)},\mathbb G_i^{(n)})(t)
=
\int_{(t_i,t]}\frac{1}{\mathbb W_i^{(n)}(s-)}\,\mathbb G_i^{(n)}(\dd s)
=
\mathbb B_i^{(n)}(t)-\mathbb B_i^{(n)}(t_i).
\end{align*}
Moreover, Proposition~\ref{prop:alt_rep} gives
\begin{align*}
\Psi_i(W_i^{\mathrm c},G_i^{\mathrm c})(t)
=
\int_{(t_i,t]}\frac{1}{W_i^{\mathrm c}(s-)}\,G_i^{\mathrm c}(\dd s)
=
\int_{(t_i,t]}\frac{1}{W_i(s-)}\,G_i(\dd s)
=
B_i(t)-B_i(t_i).
\end{align*}
Combining the last three displays proves the claim.
\end{proof}

\begin{remark}[Confidence intervals and bands]\label{remark:ci}
Let $\widehat K_i$ be a consistent estimator of the covariance function
\begin{align*}
K_i(s,t)=\operatorname{Cov}(\zeta_i(s),\zeta_i(t)).
\end{align*}
With individual data, it can be constructed as follows. For policy $p$, set
\begin{align*}
G_{i,p}^{\mathrm c}(t)
&:=
\int_{[0,t]}\1_{\{L^p<s\le R^p\}}\1_{\{X^p(s-)=i\}}\,\dd A^p(s),\\
W_{i,p}^{\mathrm c}(t)
&:=
\1_{\{L^p\le t<R^p\}}\1_{\{X^p(t)=i\}}.
\end{align*}
Then define
\begin{align*}
\widehat Z_{i,p}(t)
&:=
\int_{(\alpha_i,t]}\frac{1}{\mathbb W_i^{(n)}(s-)}
\,\dd\big(G_{i,p}^{\mathrm c}-\mathbb G_i^{(n)}\big)(s)\\
&\quad
-
\int_{(\alpha_i,t]}
\frac{W_{i,p}^{\mathrm c}(s-)-\mathbb W_i^{(n)}(s-)}
{\big(\mathbb W_i^{(n)}(s-)\big)^2}
\,\mathbb G_i^{(n)}(\dd s),
\qquad t\in[\alpha_i,\beta_i],
\end{align*}
and take
\begin{align*}
\widehat K_i(s,t)
:=
\frac1n\sum_{p=1}^n \widehat Z_{i,p}(s)\widehat Z_{i,p}(t).
\end{align*}
Pointwise confidence intervals for $B_i(t)-B_i(\alpha_i)$ are then given by
\begin{align*}
\mathbb B_i^{(n)}(t)-\mathbb B_i^{(n)}(\alpha_i)
\ \pm\
z_{1-\gamma/2}\sqrt{\frac{\widehat K_i(t,t)}{n}}.
\end{align*}
For simultaneous confidence bands on $[\alpha_i,\beta_i]$, one may estimate the quantile $c_{1-\gamma}$ of
\begin{align*}
\sup_{t\in[\alpha_i,\beta_i]}
\left|\frac{\zeta_i(t)}{\sqrt{K_i(t,t)}}\right|
\end{align*}
by multiplier simulation: draw independent variables $e_1,\ldots,e_n$ with mean zero and variance one, form
\begin{align*}
\zeta_i^*(t)
:=
\frac1{\sqrt n}\sum_{p=1}^n e_p\widehat Z_{i,p}(t),
\end{align*}
and let $\widehat c_{1-\gamma}$ be the empirical $(1-\gamma)$-quantile of
\begin{align*}
\sup_{t\in[\alpha_i,\beta_i]}
\left|\frac{\zeta_i^*(t)}{\sqrt{\widehat K_i(t,t)}}\right|.
\end{align*}
The simultaneous band is
\begin{align*}
\mathbb B_i^{(n)}(t)-\mathbb B_i^{(n)}(\alpha_i)
\ \pm\
\widehat c_{1-\gamma}\sqrt{\frac{\widehat K_i(t,t)}{n}},
\qquad t\in[\alpha_i,\beta_i].
\end{align*}
This covariance estimation requires individual data, because it uses products of policy-specific terms. This is in contrast to the estimator $\mathbb B_i^{(n)}$ itself, which can be computed from the aggregate payment and exposure processes appearing in \eqref{eq:BBhat}.
\hfill\qed
\end{remark}

\section{Numerical Example: disability insurance}\label{sec:numerical_examples}

For illustrative purposes, we artificially generate data from the micro-level model described in Ahmad and Bladt (2023), which we then use to demonstrate the estimation of macro-level payment functions. This artificial setting allows to evaluate the performance of our estimator.
The macro state space $\mathcal X=\{1,2,3\}$ distinguishes between the states  active, disabled, and dead. The state disabled has ten micro-level sub-states, so the enlarged state space is
\begin{align*}
\mathcal Z=\{(1,1),(2,1),\ldots,(2,10),(3,1)\}.
\end{align*}
  In state active and dead, micro levels are meaningless and are therefore skipped. The state $(3,1)$ is absorbing. The model setting is illustrated in Figure~\ref{fig:simulation_micro_macro}.  The process is simulated in age time on $[50,70]$ and starts in $(1,1)$ at age $50$. 

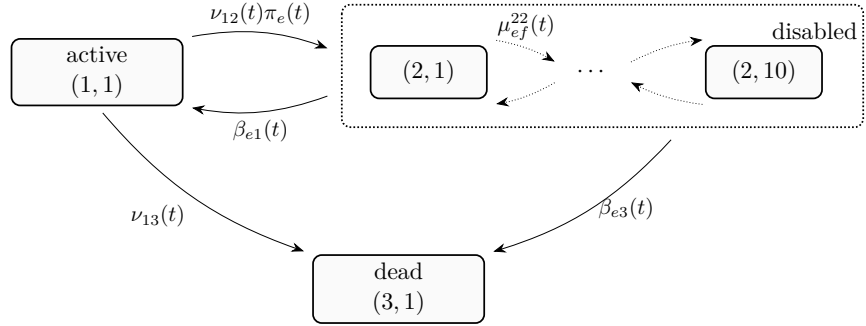
\begin{figure}[!htbp]
	\centering
	\scalebox{0.82}{
		\begin{tikzpicture}[
			node distance=2em and 0em,
			statebox/.style={
				rectangle,
				rounded corners,
				draw=black, thick,
				text width=6.5em,
				minimum height=2.4em,
				text centered,
				fill=black!2},
			microbox/.style={
				rectangle,
				rounded corners,
				draw=black, thick,
				text width=4.2em,
				minimum height=2.2em,
				text centered,
				fill=black!2},
			arr/.style={-{Stealth[length=2mm]}, shorten <=4pt, shorten >=4pt},
			darr/.style={-{Stealth[length=2mm]}, shorten <=4pt, shorten >=4pt, densely dotted}
			]
			\node[microbox] (21) {$(2,1)$};
			\node[right = 13mm of 21] (2temp) {$\cdots$};
			\node[microbox, right = 35mm of 21] (210) {$(2,10)$};
			\node[statebox, left = 30mm of 21] (11) {active\\$(1,1)$};
			\node[draw = none, fill = none, left = 48mm of 210] (test) {};
			\node[statebox, below = 28mm of test] (3) {dead\\$(3,1)$};
			\node[right = 0pt] at ($(210.south)+(0,1.1)$) {disabled};
			\draw[thick, densely dotted, rounded corners]
			($(21.north west)+(-0.45,0.65)$) rectangle ($(210.south east)+(0.65,-0.45)$);
			\path
			($(11.north east)$) edge [arr, bend left=15] node [above, font=\small] {$\nu_{12}(t)\pi_e(t)$} ($(21.north west)+(-0.55,-0.1)$)
			($(21.south west)+(-0.55,0.1)$) edge [arr, bend left=15] node [below, font=\small] {$\beta_{e1}(t)$} ($(11.south east)$)
			($(210.south)+(-1.4,-0.55)$) edge [arr, bend left=15] node [right, font=\small] {$\beta_{e3}(t)$} ($(3.north east)$)
			($(11.south)$) edge [arr, bend right=15] node [below left, font=\small] {$\nu_{13}(t)$} ($(3.north west)$)
			($(21.north east)+(0.0,0.1)$) edge [darr, bend left=10] node [above, font=\small] {$\mu_{ef}^{22}(t)$} ($(2temp.west)+(-0.1,0.1)$)
			($(2temp.east)+(0.1,0.1)$) edge [darr, bend left=10] ($(210.north west)+(0.1,0.1)$)
			($(210.south west)+(0.1,-0.1)$) edge [darr, bend left=10] ($(2temp.east)+(0.1,-0.1)$)
			($(2temp.west)+(-0.1,-0.1)$) edge [darr, bend left=10] ($(21.south east)+(0.0,-0.1)$)
			;
	\end{tikzpicture}}
	\caption{Three macro states with ten unobserved micro states in the disabled macro state.}
	\label{fig:simulation_micro_macro}
\end{figure}

The disabled block is a fitted ten-state aggregate Markov model. For $e\neq f$,
\begin{align*}
\mu_{ef}^{22}(t)=\exp(a_{ef}+b_{ef}t),
\end{align*}
and the exit rates from disabled micro state $e$ to the active and dead macro states are
\begin{align*}
\beta_{e1}(t)=\exp(c_{e1}+d_{e1}t),
\qquad
\beta_{e3}(t)=\exp(c_{e3}+d_{e3}t).
\end{align*}
Entries into disability use the fitted reset distribution
\begin{align*}
\pi_e(t)
=
\frac{\exp(\eta_e^0+\eta_e^1 t)}
{\sum_{f=1}^{10}\exp(\eta_f^0+\eta_f^1 t)},
\qquad e=1,\ldots,10,
\end{align*}
and have rates
\begin{align*}
\mu_{(1,1),(2,e)}(t)=\nu_{12}(t)\pi_e(t).
\end{align*}
The coefficients
\begin{align*}
a_{ef},\,b_{ef},\,c_{e1},\,d_{e1},\,c_{e3},\,d_{e3},\,\eta_e^0,\,\eta_e^1
\end{align*}
are the fitted coefficients from~\cite{AhmadBladt2023} and are not estimated from the payment data.
The total active-to-disabled rate is
\begin{align*}
\nu_{12}(t)
=
\begin{cases}
\exp(q(t)),
& t\le67,\\
0.0009687435,& t>67,
\end{cases}
\end{align*}
where
\begin{align*}
q(t)
&=
72.53851-10.66927t+0.53371t^2-0.012798t^3\\
&\quad
+1.4922\cdot10^{-4}t^4-6.8007\cdot10^{-7}t^5.
\end{align*}
To close the three-state model we add the simple active-to-dead rate
\begin{align*}
\nu_{13}(t)=\exp(-10.5+0.09t).
\end{align*}
This last rate is not essential for the target below; it only completes the simulated life history.

Here, we focus exclusively on disability benefits payed while the policyholder is in the disabled macro-state. We assume that
\begin{align*}
\dd A_{(2,e)}(t)= r_e\,\dd t,
\qquad
r_e=e,\quad e=1,\ldots,10,
\end{align*}
and there are no payments in the active or dead macro states. The target is the projected disabled payment function
\begin{align*}
B_2(t)-B_2(60)
=
\int_{(60,t]}
\frac{\sum_{e=1}^{10}r_e\,\Prob(X(s-)=(2,e))}
{\Prob(X(s-)\in\{(2,1),\ldots,(2,10)\})}
\,\dd s,
\qquad t\in[60,70],
\end{align*}
cf.~equation \eqref{EquilityofEdBEdA}.
There is no closed form for this quantity that is useful here, because the disabled block is a calibrated ten-state inhomogeneous Markov chain. We therefore approximate it once by an uncensored Monte Carlo sample of size $300,\!000$. This approximation is used only as a benchmark, not for estimation.

The observed samples are independently right-censored. We take $L\equiv50$ and
\begin{align*}
R=\min(70,U),\qquad U\sim \operatorname{Unif}(63,75),
\end{align*}
independently of the path. Thus all policies are observed at the beginning of the estimation window $[60,70]$, while exposure decreases toward the right endpoint.

The estimator is evaluated exactly between event times. Let
\begin{align*}
D=\{(2,1),\ldots,(2,10)\}.
\end{align*}
For $s\in[60,70]$, define the aggregate disabled exposure count and aggregate disabled payment rate by
\begin{align*}
\widehat Y(s)
&:=
\sum_{p=1}^n\1_{\{s<R^p\}}\1_{\{X^p(s)\in D\}},\\
\widehat H(s)
&:=
\sum_{p=1}^n\sum_{e=1}^{10}
r_e\1_{\{s<R^p\}}\1_{\{X^p(s)=(2,e)\}}.
\end{align*}
The target estimator then becomes
\begin{align*}
\mathbb B_2^{(n)}(t)-\mathbb B_2^{(n)}(60)
=
\int_{(60,t]}\frac{\widehat H(s)}{\widehat Y(s)}\,\dd s,
\qquad t\in[60,70],
\end{align*}
with the integrand taken as zero when $\widehat Y(s)=0$. We again remark that this formula uses only aggregate exposure counts and aggregate payment rates.

For the confidence intervals, we follow Remark~\ref{remark:ci}, where individual data are used through
\begin{align*}
\widehat Z_p(t)
=
\int_{(60,t]}
n\left\{
\frac{\widehat H_p(s)}{\widehat Y(s)}
-
\frac{\widehat Y_p(s)\widehat H(s)}{\widehat Y(s)^2}
\right\}\dd s,
\end{align*}
where
\begin{align*}
\widehat Y_p(s)
&:=
\1_{\{s<R^p\}}\1_{\{X^p(s)\in D\}},\\
\widehat H_p(s)
&:=
\sum_{e=1}^{10}r_e\1_{\{s<R^p\}}\1_{\{X^p(s)=(2,e)\}}.
\end{align*}
The same zero convention is used when $\widehat Y(s)=0$.
We estimate $\operatorname{Var}(\zeta_2(t))$ by
\begin{align*}
\frac1n\sum_{p=1}^n\widehat Z_p(t)^2
\end{align*}
and plot the corresponding pointwise $95\%$ confidence intervals.

\begin{figure}[!htbp]
\centering
\begin{minipage}[t]{0.49\textwidth}
\centering
{\scriptsize $n=500$\par\vspace{-1mm}}
\includegraphics[width=\linewidth,trim=4mm 3mm 4mm 4mm,clip]{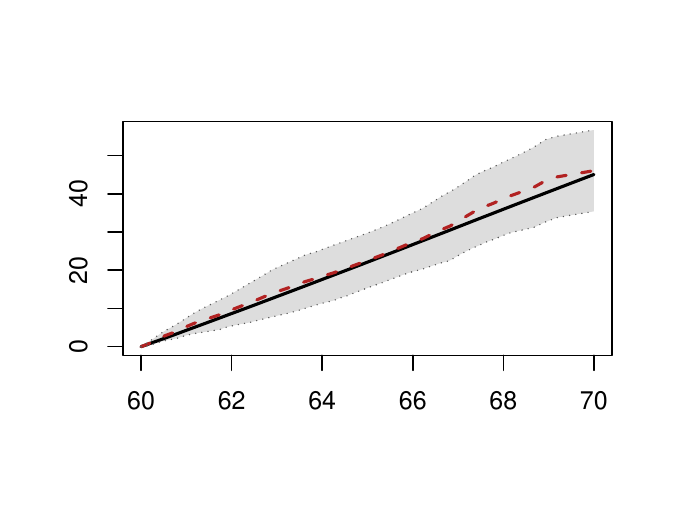}
\end{minipage}\hfill
\begin{minipage}[t]{0.49\textwidth}
\centering
{\scriptsize $n=1000$\par\vspace{-1mm}}
\includegraphics[width=\linewidth,trim=4mm 3mm 4mm 4mm,clip]{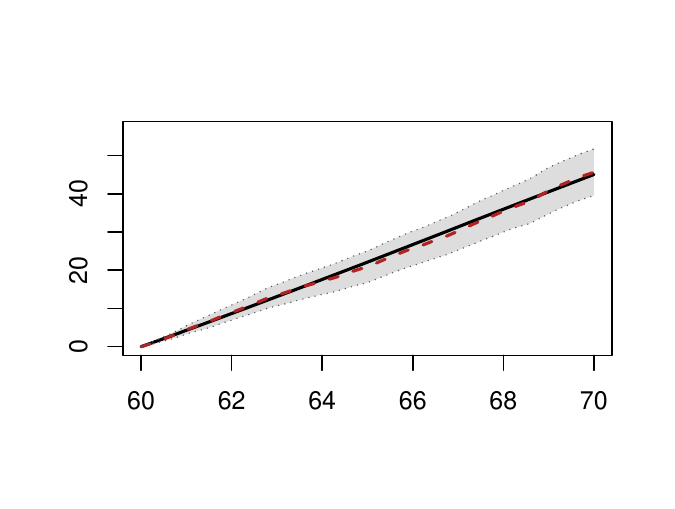}
\end{minipage}

\vspace{-1mm}

\begin{minipage}[t]{0.49\textwidth}
\centering
{\scriptsize $n=5000$\par\vspace{-1mm}}
\includegraphics[width=\linewidth,trim=4mm 3mm 4mm 4mm,clip]{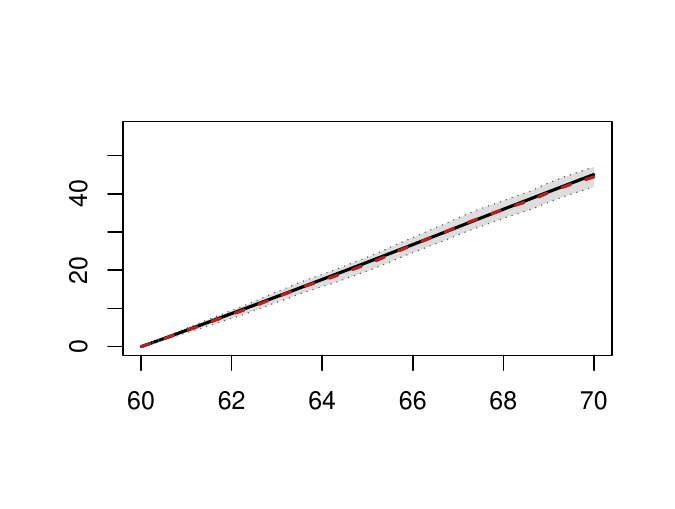}
\end{minipage}\hfill
\begin{minipage}[t]{0.49\textwidth}
\centering
{\scriptsize $n=10000$\par\vspace{-1mm}}
\includegraphics[width=\linewidth,trim=4mm 3mm 4mm 4mm,clip]{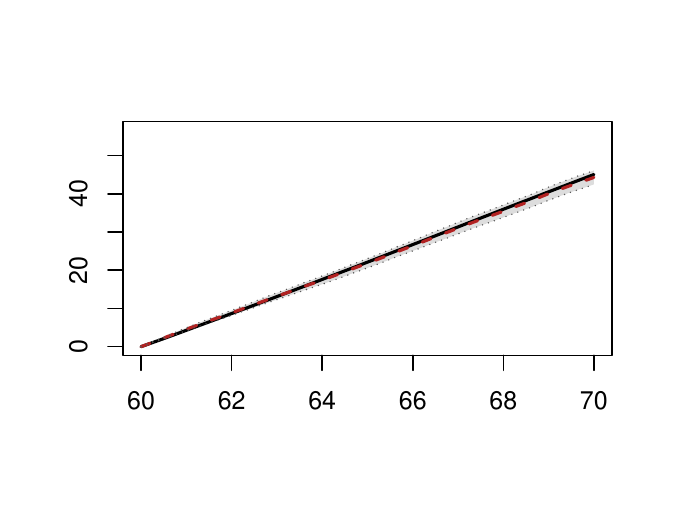}
\end{minipage}
\caption{Estimation of $B_2(t)-B_2(60)$ under independent right-censoring. The solid black curve is the Monte Carlo benchmark, the dashed red curve is the aggregate estimator, and the grey region is the pointwise $95\%$ confidence band.}
\label{fig:final_simulation}
\end{figure}

Figure~\ref{fig:final_simulation} illustrates two features of the theory. First, the estimator is based only on aggregate exposure counts and aggregate payment rates, but it still estimates the projected disabled payment function effectively, which is the primary goal of this paper. The micro process is hidden and highly structured, yet it enters the target only through the conditional average disabled payment rate. Notice that in this numeric example, the true curve is close to linear due to the payment process structure and simulation setup.

Second, the confidence bands reflect the information loss from censoring. The denominator $\widehat Y(s)$ is an exposure quantity, and it decreases toward the end of the interval because policies are censored after age $63$. The bands therefore react not only to $n$, but also to where exposure is available. The panels in Figure~\ref{fig:final_simulation} illustrate the contraction of the pointwise uncertainty and the improved accuracy of the aggregate estimator toward the benchmark at the expected root-$n$ scale. However, we stress this additional uncertainty quantification is only possible when individual information is also available.

\newpage

\bibliographystyle{apalike}
\bibliography{final.bib}

\begin{thebibliography}{}

\bibitem[Ahmad and Bladt, 2023]{AhmadBladt2023}
Ahmad, J. and Bladt, M. (2023).
\newblock Aggregate markov models in life insurance: estimation via the em
  algorithm.

\bibitem[Ahmad and Bladt, 2024]{AhmadBladt2024}
Ahmad, J. and Bladt, M. (2024).
\newblock Aggregate markov models in life insurance: estimation via the em
  algorithm.
\newblock {\em Scandinavian Actuarial Journal}, 2024(6):533--560.

\bibitem[Ahmad et~al., 2023]{AhmadBladtFurrer2023}
Ahmad, J., Bladt, M., and Furrer, C. (2023).
\newblock {Aggregate Markov models in life insurance: Properties and
  valuation}.
\newblock {\em Insurance: Mathematics and Economics}, 113:50--69.

\bibitem[Andersen and Lollike, 2024]{Andersen2024}
Andersen, I. and Lollike, A. (2024).
\newblock Efficient projections of with-profit life insurance using lumping.
\newblock {\em European Actuarial Journal}, 14:21--61.

\bibitem[Andersen et~al., 2012]{AndersenBorganGillKeiding2012}
Andersen, P., Borgan, {\O}., Gill, R., and Keiding, N. (2012).
\newblock {\em Statistical models based on counting processes}.
\newblock Springer Science \& Business Media.

\bibitem[Bladt and Furrer, 2025]{BladtFurrer2025}
Bladt, M. and Furrer, C. (2025).
\newblock {Conditional Aalen--Johansen estimation}.
\newblock {\em Scandinavian Journal of Statistics}, 52(2):873--902.

\bibitem[Furrer and Sandqvist, 2025]{FurrerSandqvist2025}
Furrer, C. and Sandqvist, O. (2025).
\newblock Loss of earning capacity in denmark--an actuarial perspective.
\newblock {\em arXiv preprint arXiv:2501.11578}.

\bibitem[Linton and Nielsen, 1995]{LintonNielsen1995}
Linton, O. and Nielsen, J.~P. (1995).
\newblock A kernel method of estimating structured nonparametric regression
  based on marginal integration.
\newblock {\em Biometrika}, pages 93--100.

\bibitem[Milbrodt and R{\"o}hrs, 2016]{MilbrodtRoehra2016}
Milbrodt, H. and R{\"o}hrs, V. (2016).
\newblock {\em Aktuarielle Methoden der deutschen Privaten
  Krankenversicherung}, volume~34.
\newblock VVW GmbH.

\bibitem[Milhaud and Dutang, 2018]{MilhaudDutang2018}
Milhaud, X. and Dutang, C. (2018).
\newblock Lapse tables for lapse risk management in insurance: a competing risk
  approach.
\newblock {\em European Actuarial Journal}, 8(1):97--126.

\bibitem[Ramlau-Hansen, 1983]{RamlauHansen1988}
Ramlau-Hansen, H. (1983).
\newblock Smoothing counting process intensities by means of kernel functions.
\newblock {\em The Annals of Statistics}, pages 453--466.

\bibitem[Reck et~al., 2025]{ReckSchuppReuss2025}
Reck, L., Schupp, J., and Reu\ss, A. (2025).
\newblock A multistate analysis of policyholder behaviour in life
  insurance—lasso-based modelling approaches.
\newblock {\em Risks}, 13(4).

\end{thebibliography}

\newpage
\appendix

\section{Empirical process results}\label{app:ep_tools}

This appendix collects the standard results used in the proofs.

We begin with the functional map underlying the inverse-weighted estimator.

\begin{lemma}[Hadamard differentiability of the ratio-integral map]\label{lem:hadamard}
Fix $0\le a<b\le T$. Define
\begin{align*}
\Psi(x,y)(t):=\int_{(a,t]}\frac{1}{x(s-)}\,y(\dd s),\quad t\in[a,b],
\end{align*}
on
\begin{align*}
\mathbb D_\eta:=\{(x,y)\in D([a,b])\times D([a,b]):\inf_{s\in[a,b]}x(s)\ge\eta,\ y\in BV([a,b])\},
\end{align*}
for fixed $\eta>0$. Then $\Psi:\mathbb D_\eta\to D([a,b])$ is Hadamard differentiable at every $(x,y)\in\mathbb D_\eta$, tangentially to $D([a,b])\times D([a,b])$, with derivative
\begin{align}
\Psi'_{(x,y)}(h_1,h_2)(t)
&=
\int_{(a,t]}\frac{1}{x(s-)}\,h_2(\dd s)
-\int_{(a,t]}\frac{h_1(s-)}{x(s-)^2}\,y(\dd s),
\label{eq:hadamard_deriv}
\end{align}
where the Stieltjes integrals are understood by integration by parts.
\end{lemma}

Continuity used in the consistency arguments follows immediately.

\begin{corollary}[Continuity of the ratio-integral map]\label{lem:map_cont}
Fix $0\le a<b\le T$ and $(x,y)\in\mathbb D_\eta$. If
\begin{align*}
(x_n,y_n)\to(x,y)\quad\text{in }D([a,b])\times D([a,b])
\end{align*}
for the uniform norm, with $(x_n,y_n)\in\mathbb D_\eta$ for all large $n$, then
\begin{align*}
\sup_{t\in[a,b]}\left|\Psi(x_n,y_n)(t)-\Psi(x,y)(t)\right|\to0.
\end{align*}
\end{corollary}

The next result is the uniform law used for monotone c\`adl\`ag component processes.

\begin{theorem}[Abstract Glivenko--Cantelli for monotone c\`adl\`ag maps]\label{th:abstract_gc_cadlag}
Let $T:[0,\theta]\to\R$ be bounded, nondecreasing, and c\`adl\`ag. Let $T^{(n)}:[0,\theta]\to\R$ be random nondecreasing c\`adl\`ag maps. If for every $t\in[0,\theta]$,
\begin{align*}
T^{(n)}(t)\as T(t),
\quad
T^{(n)}(t-)\as T(t-),
\end{align*}
then
\begin{align*}
\sup_{0\le t\le\theta}|T^{(n)}(t)-T(t)|\as0.
\end{align*}
\end{theorem}

This is the compact-interval version of the abstract Glivenko--Cantelli lemma.

\begin{lemma}[Bracketing bound for monotone cumulative classes]\label{lem:mono_bracketing}
Let $\{V_t:t\in[0,\theta]\}$ be pointwise measurable, c\`adl\`ag, and nondecreasing in $t$, and set
\begin{align*}
\mathcal F_V:=\{f_t:t\in[0,\theta]\},\quad f_t(\omega):=V_t(\omega),
\end{align*}
with envelope $F(\omega):=V_\theta(\omega)\in L^2(\Prob)$. Then for every $\varepsilon\in(0,1]$,
\begin{align*}
N_{[\,]}\left(\varepsilon\|F\|_{L^2(\Prob)},\mathcal F_V,L^2(\Prob)\right)\le \frac{2}{\varepsilon^2}+1.
\end{align*}
\end{lemma}
\begin{proof}
Define
\begin{align*}
\Gamma(t):=\E[V_tF],\quad t\in[0,\theta].
\end{align*}
Since $0\le V_t\le F$, $\Gamma$ is nondecreasing c\`adl\`ag, with
\begin{align*}
0\le \Gamma(\theta)-\Gamma(0)\le \E[F^2]=\|F\|_{L^2(\Prob)}^2.
\end{align*}
Choose a partition
\begin{align*}
0=t_0<t_1<\cdots<t_m=\theta
\end{align*}
such that
\begin{align*}
\Gamma(t_r-)-\Gamma(t_{r-1})\le \varepsilon^2\|F\|_{L^2(\Prob)}^2,\quad r=1,\dots,m,
\end{align*}
and
\begin{align*}
m\le 2\varepsilon^{-2}+1.
\end{align*}
For each $r$, define brackets
\begin{align*}
l_r:=f_{t_{r-1}},\quad u_r:=f_{t_r-}.
\end{align*}
Monotonicity in $t$ implies that for every $t\in[t_{r-1},t_r)$,
\begin{align*}
l_r\le f_t\le u_r.
\end{align*}
Also, $0\le u_r-l_r\le F$, hence
\begin{align*}
\|u_r-l_r\|_{L^2(\Prob)}^2
&=\E[(u_r-l_r)^2]
\le \E[(u_r-l_r)F]\\
&=\Gamma(t_r-)-\Gamma(t_{r-1})
\le \varepsilon^2\|F\|_{L^2(\Prob)}^2.
\end{align*}
Therefore each bracket has $L^2$-width at most $\varepsilon\|F\|_{L^2(\Prob)}$, and at most $2\varepsilon^{-2}+1$ brackets are needed.
\end{proof}

\begin{theorem}[Bracketing CLT, Donsker's theorem]\label{Donsker}
A measurable class $\mathcal{F}$ with envelope $F\in L^2(\Prob)$ is $\Prob$-Donsker if
\begin{align*}
J_{[\,]}(1,\mathcal F,L^2(\Prob))
:=
\int_{(0,1)}
\sqrt{\log N_{[\,]}(\varepsilon,\mathcal F,L^2(\Prob))}
\,\dd\varepsilon
<\infty.
\end{align*}
\end{theorem}

\begin{lemma}[From marginal to joint weak convergence]\label{stack_convergence}
Let $X_n\dist X$ in $\ell^\infty(T)$ and $Y_n\dist Y$ in $\ell^\infty(S)$, with tight limits. If all mixed finite-dimensional vectors of $(X_n,Y_n)$ converge to those of $(X,Y)$, then $(X_n,Y_n)\dist (X,Y)$ in $\ell^\infty(T)\times\ell^\infty(S)$.
\end{lemma}

\begin{theorem}[Functional $\delta$-method]\label{delta_method}
Let $\mathbb D$ and $\mathbb E$ be normed spaces. Let $\phi\colon\mathbb D_{\phi}\subseteq \mathbb D\to \mathbb E$ be Hadamard differentiable at $\theta$ tangentially to $\mathbb D_{0}$. Let $X_n\colon\Omega\to \mathbb D_{\phi}$ be maps with
\begin{align*}
r_n(X_n-\theta)\dist X
\end{align*}
for some sequence of constants $r_n\to \infty$, where $X$ is tight and takes values in $\mathbb D_{0}$. Then
\begin{align*}
r_n\big(\phi(X_n)-\phi(\theta)\big)\dist \phi'_{\theta}(X)
\end{align*}
in $\mathbb E$.
\end{theorem}

\end{document}